\documentclass[final, 3p]{elsarticle}

 \usepackage{shortcuts,multirow,graphicx}
\newproof{pf}{Proof}
 
 \usepackage[latex={-interaction=nonstopmode},crop=off,off]{auto-pst-pdf}
  \usepackage{psfrag}
\usepackage{booktabs}
\journal{Applied and Computational Harmonic Analysis}

\DeclareFontFamily{U}{mathx}{\hyphenchar\font45}
\DeclareFontShape{U}{mathx}{m}{n}{<-> mathx10}{}
\DeclareSymbolFont{mathx}{U}{mathx}{m}{n}
\DeclareMathAccent{\widebar}{0}{mathx}{"73}

\providecommand{\w}{{\omega}}
\providecommand{\Eq}{\:=\:}
\providecommand{\Leq}{\:\leq\:}

\providecommand{\Dwc}{\Delta_{\w_c}^2}
\providecommand{\Dwl}{\Delta_{\w_\ell}^2}
\providecommand{\Dwp}{\Delta_{\w_p}^2}
\providecommand{\Dn}{\Delta_n^2}
\providecommand{\etac}{{\eta_c}}
\providecommand{\etal}{{\eta_\ell}}
\providecommand{\etap}{{\eta_p}}

\providecommand{\Zspace}{\mathbb{Z}}

\providecommand{\norm}[1]{\left\lVert#1\right\rVert}
\renewcommand{\abs}[1]{\left\lvert#1\right\rvert}

\makeatletter
\renewcommand*\env@matrix[1][\arraystretch]{%
  \edef\arraystretch{#1}%
  \hskip -\arraycolsep
  \let\@ifnextchar\new@ifnextchar
  \array{*\c@MaxMatrixCols c}}
\makeatother

\begin{document}
\begin{frontmatter}
\suppressfloats

\title{Sequences with Minimal Time-Frequency Uncertainty}
\tnotetext[t1]{This work was supported by ERC Advanced Grant--Support for Frontier Research--SPARSAM, Nr: 247006.}
\author{Reza Parhizkar\corref{cor1}}
\ead{reza.parhizkar@epfl.ch}

\author{Yann Barbotin}
\ead{yann.barbotin@epfl.ch}

\author{Martin Vetterli}
\ead{martin.vetterli@epfl.ch}

\cortext[cor1]{Corresponding author, Phone number: +41-21-693-5636}

\address{School of Computer and Communication Sciences\\
Ecole Polytechnique F\'{e}d\'{e}rale de Lausanne (EPFL), CH-1015 Lausanne, Switzerland}

\begin{abstract}
A central problem in signal processing and communications is to design signals
that are compact both in time and frequency. Heisenberg's uncertainty principle states that a
given function cannot be arbitrarily compact both in time and frequency, defining an
``uncertainty" lower bound. Taking the variance as a measure of localization in time and frequency, Gaussian functions reach this
bound for continuous-time signals. For sequences, however, this is not true; it is known that Heisenberg's bound is
generally unachievable. For a chosen frequency variance, we formulate the search for ``maximally compact
sequences'' as an exactly and efficiently solved  convex optimization problem, thus
providing a sharp uncertainty principle for sequences. Interestingly, the optimization formulation also reveals
that maximally compact sequences are derived from Mathieu's harmonic cosine
function of order zero. We further provide rational asymptotic
expansions of this sharp uncertainty bound. We use the derived bounds as a benchmark to compare the compactness of well-known window functions with that of the optimal Mathieu's functions. 
\end{abstract}

\end{frontmatter}
\section{Introduction}
\label{sec:introduction}

Suppose you are asked to design filters that are sharp in the frequency domain and at the same time compact in the time domain. The same problem is posed in designing sharp probing basis functions with compact frequency characteristics. In order to formulate these problems mathematically, we need to have a correct and universal definition of compactness and clarify what we mean by saying a signal is spread in time or frequency. 

These notions are well defined and established for continuous-time signals \cite{heisenberg1927,vetterli} and their properties are studied thoroughly in the literature. For such signals, we can define the time and frequency characteristics of a signal as in Table \ref{tab:TF-cont}. Note the connection of these definitions with the mean and variance of a probability distribution function $|x(t)|^2 /\norm{x}^2$. The value of $\Delta_t^2$ is considered as the spread of the signal in the time domain while $\Dwc$ represents its spread in the frequency domain. We say that a signal is compact in time (or frequency) if it has a small time (or frequency) spread. 

\begin{table*}[tb]
\centering
\begin{tabular}{@{}p{2cm}  p{6cm}  l@{}}
\toprule[1.2pt]
\textbf{domain} & \multicolumn{1}{l}{\textbf{center}} & \multicolumn{1}{l}{\textbf{spread}}  \\
\midrule
\addlinespace[8pt]
 \textbf{time} & $\mu_t = \frac{1}{\norm{x}^2}\int_{t\in \R}t |x(t)|^2 dt$ & $\Delta_t^2 = \frac{1}{\norm{x}^2}\int_{t\in \R} (t - \mu_t)^2|x(t)|^2 dt$\\\\
 \textbf{frequency}$\,\quad\,$ & $\mu_{\w_c} = \frac{1}{2\pi\norm{x}^2}\int_{\w\in \R} \w |X(\w)|^2 d\w\quad$ & $\Dwc = \frac{1}{2\pi\norm{x}^2}\int_{\w \in \R} (\w - \mu_\w)^2 |X(\w)|^2 d\w$\\\addlinespace[5pt]
  \bottomrule[1.2pt]
\end{tabular}
\caption{Time and frequency centers and spreads for a continuous time signal $x(t)$.}
\label{tab:TF-cont}
\end{table*}

The Heisenberg uncertainty principle \cite{heisenberg1927, robertson1929, Schrodinger1930} states that continuous-time
signals cannot be arbitrarily compact in both domains. Specifically,
for any $x(t) \in L^2(\R)$, 
\begin{equation}\label{eq:introhupc}
\etac = \Delta_t^2 \,\Dwc \geq \frac{1}{4}\,,
\end{equation}
where the lower bound is achieved for Gaussian signals of the form $x(t) = \gamma e^{-\alpha (t-t_0)^2+j\omega_0t}, \, \alpha >0$ \cite{gabor46}. The subscript $c$ stands for continuous-time definitions. We call $\etac$ the \textit{time-frequency spread} of $x$.

Although the continuous Heisenberg uncertainty principle is widely used in theory, in practice we often work with discrete-time signals (e.g.~filters and wavelets). Thus, equivalent definitions for discrete-time sequences are needed in signal processing. In the next section we study two common definitions of center and spread available in the literature. 

\subsection{Uncertainty principles for sequences}
\label{subsec:background}

An obvious and intuitive extension of the definitions in Table \ref{tab:TF-cont} for discrete-time signals is presented in Table \ref{tab:TF-disc-1}, where
\begin{equation}
X(e^{j\w}) = \sum_{k\in\Z} x_k e^{-j\w k}\, \quad \w \in \R\,,
\label{eq:DTFT_def}
\end{equation}
is the discrete-time Fourier transform (DTFT) of $x_n$.

\begin{table*}[tb]
\centering
\begin{tabular}{@{}p{2cm}  p{6cm}  l@{}}
\toprule[1.2pt]
\textbf{domain} & \multicolumn{1}{l}{\textbf{center}} & \multicolumn{1}{l}{\textbf{spread}}  \\
\midrule
\addlinespace[8pt]
 \textbf{time} & $\mu_n = \frac{1}{\norm{x}^2}\sum_{k\in\Z}k |x_k|^2$ & $\Dn = \frac{1}{\norm{x}^2}\sum_{k\in\Z} (k - \mu_n)^2|x_k|^2$\\\\
\textbf{frequency}$\,\quad\,$ & $\mu_{\w_{\ell}} = \frac{1}{2\pi\norm{x}^2}\int_{-\pi}^\pi \w |X(e^{j\w})|^2 d\w\quad$ & $\Dwl = \frac{1}{2\pi\norm{x}^2}\int_{-\pi}^\pi (\w - \mu_\w)^2 |X(e^{j\w})|^2 d\w$\\\addlinespace[5pt]
 \bottomrule[1.2pt]
\end{tabular}
\caption{Time and frequency centers and spreads for a discrete time signal $x_n$ as extensions of Table \ref{tab:TF-cont} \cite{vetterli}.}
\label{tab:TF-disc-1}
\end{table*}

Using the definitions in Table \ref{tab:TF-disc-1}  \cite{vetterli}, we
can also state the Heisenberg uncertainty principle for discrete-time
signals as 
\begin{equation}\label{eq:introhupl}
\etal = \Dn \,\Dwl > \frac{1}{4}\,,\quad  x_n \in \ell^2(\Z)\text{ with }X(e^{j\pi}) = 0\,,
\end{equation}

where the subscript $\ell$ stands for \emph{linear} in reference to the
definition of the frequency spread. Note the extra assumption on the Fourier transform of the signal in \eqref{eq:introhupl}. This assumption is necessary for the result to hold.

\begin{Example}
Take $x_n = \delta_n + 7 \delta_{n-1} + 2 \delta_{n-2}$. It is easy to verify that that $\abs{X(e^{j\pi})}=0.22\neq 0$, which violates the condition $X(e^{j\pi})=0$. The linear time-frequency spread of this signal according to Table \ref{tab:TF-disc-1} is $\etal = 0.159 < 1/4$. 
%
\end{Example}
In addition to the restriction on the Heisenberg uncertainty
principle, the definitions in Table \ref{tab:TF-disc-1} do not capture
the periodic nature of $X(e^{j\w})$ for the frequency center and
spread. In the search for more natural properties, we can adopt
definitions for circular moments widely used in quantum mechanics
\cite{Breitenberger1985} and directional statistics \cite{directionalMardia}.

\begin{definition}For a sequence $x_n, \, n\in\Z$, with a $2\pi$-periodic DTFT, $X(e^{j\w})$ as in \eqref{eq:DTFT_def}, 
the \textit{first trigonometric moment} is defined as \citep{Prestin1999, Prestin2003}
\begin{equation}
\begin{aligned}
\tau(x) & = \frac{1}{2\pi\norm{x}^2}\int_{-\pi}^{\pi} e^{j\w}|X(e^{j\w})|^2 d\w\\
&\overset{(a)}{=} \frac{1}{\norm{x}^2}\sum_{k\in \Z} x_k\, x^*_{k+1}\,,
\end{aligned}
\label{eq:tau}
\end{equation}
where $(a)$ follows from Parseval's equality. 
\end{definition}
The first trigonometric moment was originally defined for probability
distributions on a circle. With proper normalization, this definition
applies also to periodic functions.
 
\begin{definition}Using \eqref{eq:tau}, the \textit{periodic frequency spread} is defined as \cite{Breitenberger1985}:
\begin{equation}
\Dwp = \frac{1 - |\tau(x)|^2}{|\tau(x)|^2} = \left|\frac{\norm{x}^2}{\sum_{k\in\Z}x_k\, x_{k+1}^*}\right|^2-1\,,
\label{eq:deltawdisc}
\end{equation}
\end{definition}

where $\tau(x)$ is defined in \eqref{eq:tau}. This definition makes only sense when $\tau(x) \neq 0$. If $\tau(x) = 0$, we set $\Dwp = \infty$. Figure \ref{fig:dwp} illustrates pictorially how this definition corresponds to the first trigonometric moment of the periodic signal $X(e^{j\omega})$ (The figure is inspired from \cite{erb09}).

The definition of $\Dn$ remains unchanged as in Table \ref{tab:TF-disc-1}. These definitions are summarized in Table \ref{tab:TF-disc-2}. Using these definitions, Breitenberger \cite{Breitenberger1985} states the uncertainty relation for sequences as
\begin{equation}
\label{eq:hupdl}
\eta_p =  \Dn \ \Dwp \geq
\frac{1}{4}\ ,\quad \text{for }\|x\|_0>1\,.
\end{equation}
The condition $\norm{x}_0 > 1$ avoids the case $\Dn = 0$, which would happen for $x_n = \gamma \delta_{n-n_0}$.  
\begin{figure}[tb]
\centering
\includegraphics[width=0.6\textwidth]{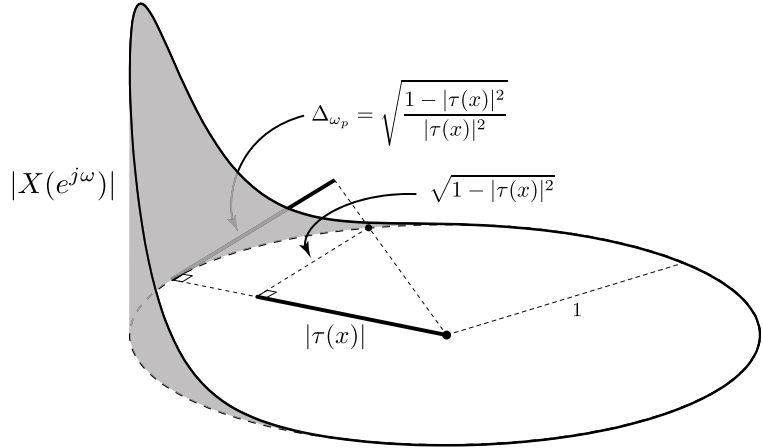}
\caption{$X(e^{j\omega})$ is a periodic function defined on the unit circle. The figure shows the correspondence between $\tau(x)$ and the periodic frequency spread of the signal.}
\label{fig:dwp}
\end{figure}

%
\begin{table*}[tb]
\centering
\begin{tabular}{@{}p{2.5cm}  p{5cm}  l@{}}
\toprule[1.2pt]
\textbf{domain} & \multicolumn{1}{l}{\textbf{center}} & \multicolumn{1}{l}{\textbf{spread}}  \\
\midrule
\addlinespace[8pt]
 \textbf{time} & $\mu_n = \frac{1}{\norm{x}^2}\sum_{k\in\Z}k |x_k|^2$ & $\Dn = \frac{1}{\norm{x}^2}\sum_{k\in\Z} (k - \mu_n)^2|x_k|^2$\\\\
\textbf{frequency}$\,\quad\,$ & $\mu_{\w_{p}} = 1-\tau(x)
\quad\quad$ & $\Dwp = \frac{1-\abs{\tau(x)}^2}{\abs{\tau(x)}^2} = \left|\frac{\norm{x}^2}{\sum_{k\in\Z}x_k\, x_{k+1}^*}\right|^2-1$\\\addlinespace[5pt]
\bottomrule[1.2pt]
\end{tabular}
\caption{Time and frequency centers and spreads for a discrete time
  signal $x_n$ using circular moments, where $\tau(x)$ is defined in \eqref{eq:tau}.}
\label{tab:TF-disc-2}
\end{table*}

\subsection{Contribution}
\label{subsec:contribution}

\begin{figure}[tb]
\centering
  %
\includegraphics[width = 0.9 \linewidth]{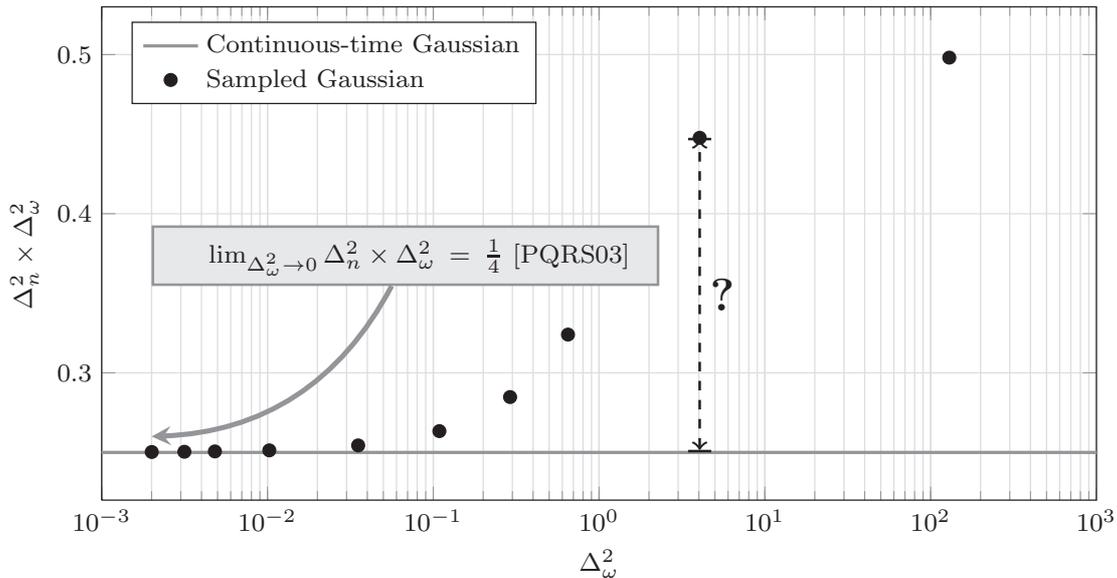}
\caption{\textbf{Time-frequency spread of continuous vs. sampled Gaussians.} The solid line shows the $1/4$ Heisenberg bound which is achieved by continuous Gaussian signals (with definitions in Table \ref{tab:TF-cont}), while the markers show the time-frequency spread (according to Table \ref{tab:TF-disc-2}) of sampled Gaussian sequences. The question is if the gap between the two curves is inherited from the properties of sequences or sampled Gaussians are not optimal in discrete domain. } 
\label{fig:cont-vs-disc-gaus}
  \end{figure}

In this paper, using the definitions in Table
\ref{tab:TF-disc-2}, we revisit the Heisenberg uncertainty principle for discrete-time signals. We address the fundamental yet unanswered question: If someone asks us to design a discrete filter
with a certain frequency spread ($\Dwp$ fixed), can we return the sequence
with minimal time spread $\Dn$? In other words, the problem is to find the solution to 
\begin{equation}
\label{eq:opt1}
\begin{aligned}
\Delta_{n,\text{opt}}^2 \ = \ &\underset{x_{n}}{\text{minimize}}
& & \Dn\\
& \text{subject to}
& &  \Dwp \ = \ \sigma^2 (\text{fixed})\,.
\end{aligned}
\end{equation}
In order to provide an insight on the uncertainty principle in the discrete-time domain, we do a simple test. In Figure \ref{fig:cont-vs-disc-gaus} we show the time-frequency spread of continuous Gaussian signals by the solid line on $1/4$. This is the Heisenberg's uncertainty bound which is achievable by continuous-time Gaussians. Further, using the definitions in Table \ref{tab:TF-disc-2}, we compute the time-frequency spread of sampled Gaussian sequences. According to Prestin et al.~\cite{Prestin2003}, the time-frequency spread tends to $1/4$ as the frequency spread of Gaussians decreases. However, when the frequency spread is large, the values are far from the uncertainty bound. Two questions arise here; is the $1/4$ uncertainty bound also tight for discrete sequences? and are sampled Gaussians the minimizers of the uncertainty in the discrete domain? Answers to these questions will be apparent if we can solve \eqref{eq:opt1}. 

\begin{definition}We call the solution of \eqref{eq:opt1} a \textit{maximally compact} sequence.
\end{definition}

Framing the design of maximally compact sequences as an optimization
problem, we show that contrary to the continuous case, it is not
possible to reach a constant time-frequency lower bound for arbitrary
time or frequency spreads. We further develop a simple optimization
framework to find maximally compact sequences in the time domain for a
given frequency spread. In other words, we provide in a constructive and numerical way,
a \emph{sharp uncertainty principle for sequences}, later shown in Figure~\ref{fig:bounds}. We also show that the Fourier spectra of maximally compact sequences are in fact a very special class of Mathieu's functions. Using the asymptotical expansion of these functions, we develop closed-form bounds on the time-frequency spread of maximally compact sequences.

\subsection{Related Work}
\label{subsec:related-work}
\begin{sloppypar}
The classic uncertainty principle \cite{heisenberg1927} assumes
continuous-time non-periodic signals. Several works in the signal processing community also address the discrete-time/discrete-frequency case \cite{Matolcsi1973, Donoho1989, Przebinda1999, Massar2008, Krahmer2008, Ghobber2011}. Our work bridges these two
cases by considering the discrete-time/continuous-frequency regime. 

Note that not all studies about the uncertainty principle concern the notion of spread. For example, the authors in \cite{Przebinda1999} propose the uncertainty bound on the information content of signals (entropy) and \cite{Donoho1989} provides a bound on the non-zero coefficients of discrete-time sequences and their discrete Fourier transforms. 
\end{sloppypar}
 The discrete-time/continuous-frequency scenario has been recently encountered in
 many practical applications in signal processing. Examples include uncertainty principle on graphs
 \cite{agaskar2012}, on spheres \cite{Khalid2012}, and on Riemannian manifolds \cite{erb09}. Studies on the periodic frequency
 spread can be found  in \cite{Breitenberger1985} and \cite{vonMises1918}. The most
 comprehensive work on the uncertainty relations for discrete sequences is found in \cite{Prestin2003}. The authors show 
that $1/4$ is a lower-bound on the time-frequency spread,
 which can only be achieved asymptotically as the sequence spreads in time. We provide sharp achievable bounds in the non-extreme case which match the results in \cite{Prestin2003} in the asymptotic regime.

This problem is similar to---although different than---the design of Slepian's Discrete Prolate Spheroidal Sequences (DPSS's). First introduced by Slepian in 1978 \cite{slepian78}, DPSS's are sequences designed to be both limited in the time and band-limited in the frequency domains. For a finite length, $N$ in time and a cut-off frequency $W$, the DPSS's are a collection of $N$ discrete-time sequences that are strictly band-limited to the digital frequency range $|f|<W$, yet highly concentrated in time to the index range $n = 0, 1, \cdots, N-1$. Such sequences can be found using an algorithm similar to the Papoulis-Gerchberg method. Note the difference of such sequences to the ones that we intend to design in our work; we do not impose any constraints on the bandwidth of the sequences in the frequency domain. Also, the ideas presented here are applicable both to finite and infinite length sequences. Moreover, we focus on the concentration of the signals in the time and frequency domain using the notion of variance. 

\section{Main Results}
\label{sec:main-results}
The following theorem is the core of the results presented in this paper. 

\begin{theorem}
\label{thm:SDR}
For finding unit norm maximally compact sequences, it is sufficient to solve the following semi-definite program (SDP)
\begin{equation}
\begin{aligned}
& \underset{\X}{\text{minimize}}
& & \text{tr}(\A\X) \\
& \text{subject to}
& & \text{tr}(\B\X)  = \alpha\\
&&& \text{tr}(\X) = 1,\ \ \X \succeq 0\,,
\end{aligned}
\label{eq:SDR}
\end{equation}
where $\alpha = \frac{1}{\sqrt{1+\sigma^2}}$ with $\sigma^2$ the fixed periodic frequency spread. Further, $\X^{\text{opt}}$, the solution to \eqref{eq:SDR} has rank one and $\X^{\text{opt}} = \x^{\text{opt}}\ {\x^{\text{opt}}}^T$, with $\x^{\text{opt}}$ the solution of  \eqref{eq:opt1}. Matrices $\A$ and $\B$ are defined as
\begin{equation}
\A = {\setlength{\arraycolsep}{3pt} \begin{bmatrix}
\ddots	&	 	&	 	&	 	&	 	&	 	&\bm{0}	 \\
 		&	2^2	&	 	&	 	&	 	& 		& 	\\
 		&	 	&	1^2	&	 	&	 	&	 	&	 \\
 		&	 	&	 	&	0	&	 	&	 	&	 \\
 		&	 	&	 	&	 	&	1^2	&	 	&	 \\
 		&		&	 	&	 	&	 	&	2^2	&	 \\
\bm{0}	&	 	&	 	&	 	&	 	&	 	&	\ddots
\end{bmatrix}},
\quad \quad
\B = \begin{bmatrix}[1.5]
\ddots	&	\ddots	&			&			&\bm{0}\\
\ddots	&	0	 	& \frac{1}{2}	&			&			\\
		&	\frac{1}{2}	&	0		&	\frac{1}{2}	&			\\
		&			&	\frac{1}{2}	&	0		&	\ddots	\\
\bm{0}&			&			&	\ddots		&	\ddots	
\end{bmatrix}\,.
\label{eq:AandB}
\end{equation}
\end{theorem}
\begin{proof} See Section \ref{sec:proof-SDR}.\end{proof}

\begin{Remark} 
\label{rem:feasible}
The SDP \eqref{eq:SDR} is feasible; indeed we can always find a signal $\bm{x}$ with a certain frequency spread and norm one. Using this signal we can construct $\bm{X} = \bm{x}\bm{x}^T$ which shows the feasibility of this problem. 
\end{Remark}
The SDP in \eqref{eq:SDR} can be solved to an arbitrary precision by using 
existing approaches in the optimization literature; for example using
the \textsf{cvx} software package \cite{Grant2010}. This gives a constructive way to design sequences that are maximally compact in the time domain with a given frequency spread. 

\begin{Example}
\label{ex:SDP}
Take $\sigma^2 = 0.1$ to be the fixed and given frequency spread of the sequence. We can use \textsf{cvx} \cite{Grant2010} to solve the semi-definite program \eqref{eq:SDR} and find the optimal value of $\Dn = 2.62$. This results in the time-frequency spread of $\etap = 0.262$. The simple code in MATLAB is:
\begin{small}
\begin{verbatim}
cvx_begin
variable X(n,n);
minimize(trace(T*X))
subject to 
    trace(J0*X) == 1/sqrt(1+0.1)
    trace(X) == 1
    X == semi-definite(n)
cvx_end;
\end{verbatim}
\end{small}
Note that contrary to continuous-time signals, we cannot reach the $0.25$ lower bound for sequences. The resulting sequence and its DTFT are shown in Figure \ref{fig:example_SDP}. 
\begin{figure}[tb]
\centering
\includegraphics[width=\textwidth]{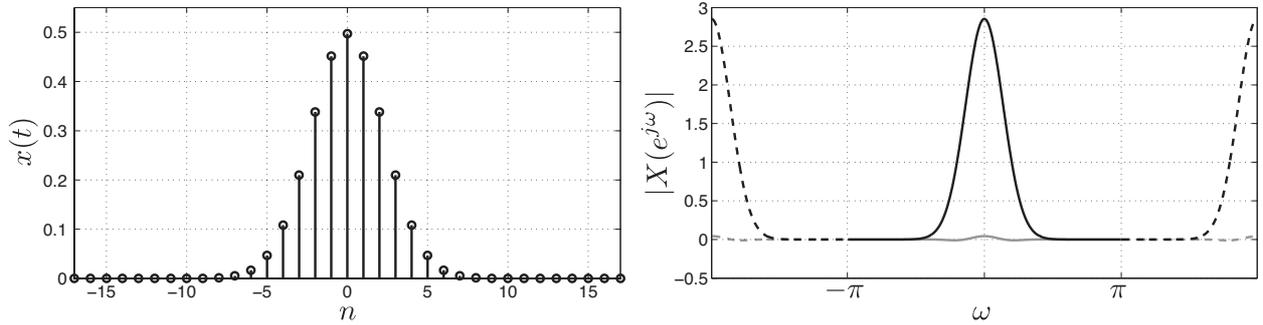}
\caption{\textbf{An example solution of \eqref{eq:SDR}.} The output of the SDP in \eqref{eq:SDR} with $\sigma^2 = 0.1$ using \textsf{cvx} in Example \ref{ex:SDP}. The gray curve in the frequency domain shows the difference from a fitted periodic Gaussian signal. The optimal value for $\Dn$ is found to be $2.62$ which results in a time-frequency spread of $\etap = 0.262$.}
\label{fig:example_SDP}
\end{figure}
\end{Example}

Note that the dual of the SDP \eqref{eq:SDR} is \cite[p. 265]{boyd2004}:
\begin{equation}
\begin{aligned}
& \underset{\lambda_1, \lambda_2}{\text{maximize}}
& & \alpha\,\lambda_1 + \lambda_2 \\
& \text{subject to}
& & \A - \lambda_1\,\B - \lambda_2\,\bm{I} \succeq 0\,.
\end{aligned}
\label{eq:dual}
\end{equation}
We will use the formulation of the dual problem many times in the rest of the paper. 

\begin{Lemma}
\label{lem:monotone}
For maximally compact sequences, $\Dn(x)$ changes monotonically with $\Dwp(x)$. 
\end{Lemma}
\begin{proof}
The feasible region of the dual \eqref{eq:dual} is shown in Figure \ref{fig:dual-feasible}. We can write \eqref{eq:dual} as
\begin{equation}
\begin{aligned}
& \underset{\lambda_1, \lambda_2}{\text{maximize}}
& & c\\
& \text{subject to}
& &\lambda_2 = c - \alpha\,\lambda_1\,,\\
& & & \A - \lambda_1\,\B - \lambda_2\,\bm{I} \succeq 0\,.
\end{aligned}
\label{eq:dual2}
\end{equation}
Note that $\alpha$ changes between $0$ and $1$ (see Figure \ref{fig:dual-feasible}). For a fixed $\alpha$, the maximum $c^{\text{opt}}$ is found by elevating the corresponding line $\lambda_2 = c-\alpha\lambda_1$ until it supports the feasible set (it is tangent to it). Since the feasible set is convex, as $\alpha$ grows (which means $\Dwp$ decreases), we need a higher elevation of the line to support the convex set, thus $c^{\text{opt}}$ (equivalently $\Dn$) increases. This phenomenon is later confirmed by the simulation results in Figure \ref{fig:bounds}.
\begin{figure}[tb]
\centering
\includegraphics[width=0.45\textwidth]{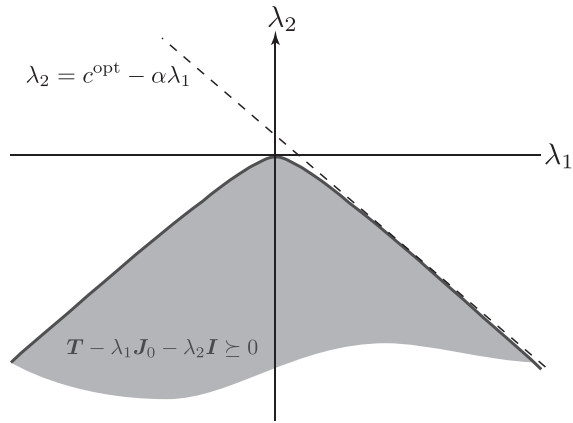}
\caption{\textbf{The feasible set of the dual problem \eqref{eq:dual} and the supporting line}. As $\alpha$ increases, we need to elevate the line more to support the feasible set, which means that the optimal value of $\Dn$ increases.}
\label{fig:dual-feasible}
\end{figure}
\end{proof}

Although Theorem \ref{thm:SDR} provides a constructive way
for finding maximally compact sequences, it does not specify the closed
form for these sequences. One would be interested to see if---in
analogy to continuous-time---sampled gaussians are maximally
compact? The answer is negative, as shown by the following theorem:

\begin{theorem}
\label{thm:mathieu}
The DTFT spectra, $X(e^{j\omega})$ of maximally compact sequences are Mathieu's functions. More specifically, 
\begin{equation}
  X(e^{j\w})\Eq \gamma_0\cdot \ce_0(-2\lambda_1 \:;\:(\w-\w_0)/2)e^{j\mu\w}\ ,
\end{equation}
where $|\gamma_0|=\|\ce_0(-2\lambda_1\:;\:(\w-\w_0)/2)\|^{-1}$,
$\w_0$ and $\mu$ are shifts in frequency or time and $\lambda_1$ is the
optimal solution of the dual problem \eqref{eq:dual}. $ \ce_0(q\ ; \w)$
is Mathieu's  harmonic cosine function of order zero. 
\end{theorem}

Mathieu's functions---widely used in quantum mechanics \cite{um2002}---are the solutions to \emph{Mathieu's differential equation} \cite[\S20.1.1]{abramowitz}:
\begin{equation}
\label{eq:mathieueq}
\frac{\partial^2 y(\omega)}{\partial \omega^{2}}+(a-2q\cos(2\omega))\cdot y(\omega)\Eq 0\,.
\end{equation}
These functions assume an even (Mathieu's cosine function) and odd form (Mathieu's sine function). For some specific pairs $(a, q)$, Mathieu's functions can be restricted to be $2\pi$ periodic. Mathieu's harmonic Cosine functions of order $m$ are thus: 
\begin{align}
\label{eq:mathieufun}
\ce_m(q\:;\:\w) \Eq \ce(a_m(q),q\:;\:\w),\ m\in\mathbb N.
\end{align}
For the proof of Theorem \ref{thm:mathieu} and further insights on Mathieu's functions, we refer the reader to Section \ref{sec:proof-mathieu}. 

Using the constructive method presented in Theorem \ref{thm:SDR}, we
can find the achievable (and tight) uncertainty principle bound for
discrete sequences. This is shown and discussed more in Section
\ref{sec:simulation-results} and Figure \ref{fig:bounds}. However, a
numerically computed boundary may not always be practical, and even
though the numerical solution exactly solves the problem, its accuracy
may be challenged. Therefore, we characterize the asymptotic behavior
of the time-frequency bound:

\begin{theorem}
\label{thm:lower-bound}
If $x_n$ is maximally compact for a given $\Dwp = \sigma^2$,  then 
\begin{equation}
\eta_p = \Dn\ \Dwp \geq \sigma^2 \left(1-\sqrt{\frac{\sigma^2}{1+\sigma^2}}\right)\,.
\label{eq:lower-bound}
\end{equation}

Further, for small values of $\sigma^2$, maximally compact sequences satisfy
\begin{equation}
\eta_p = \Dn\ \Dwp \leq \frac{\sigma^2}{8}\left(\frac{\sqrt{1+\sigma^2}}{\sqrt{1+\sigma^2}-1}-\frac{1}{2}\right)\,.
\label{eq:upper-bound}
\end{equation}
\end{theorem}
\begin{proof}
The proof for this theorem is provided in Section \ref{sec:proof_low_up}.
\end{proof}
This fundamental result states that for a given frequency spread, we
cannot design sequences which achieve the classic Heisenberg uncertainty bound. We will see how this curve compares to the classic Heisenberg bound in Section \ref{sec:simulation-results}. 

The lower bound in \eqref{eq:lower-bound} converges to $1/2$ as the
value of $\sigma^2$ grows, and  ``pushes up'' the time-frequency
spread of maximally compact sequences towards $1/2$ which is also an asymptotic upper
bound on the time-frequency spread as $\Dwp\rightarrow\infty$; indeed, one
may construct the unit-norm
sequence $x^{(\varepsilon)}_n = \varepsilon\ \delta_{n+1}
+\sqrt{1-2\varepsilon^2}\ \delta_n + \varepsilon\ \delta_{n-1}$, which
verifies $\lim_{\varepsilon\rightarrow 0}\eta_{p}(x^{(\varepsilon)})=1/2$.

On the other hand, for small values of $\sigma^2$, the upper bound in
\eqref{eq:upper-bound} converges from above to $1/4$, thus ``pushing down" 
the time-frequency spread of maximally compact sequences towards the
Heisenberg uncertainty bound $1/4$ from above. 

\subsubsection{Finite-Length Sequences}
\label{subset:finite-sequences}
The theory that we have provided so far holds for infinite
sequences. For computational purposes, we have to assume
finite length for the sequences in the time domain, which is not an
issue if the sequence length is chosen to be long enough. As a side benefit,
a length constraint on the sequence may be added without changing the design algorithm. 

\section{Proof of Theorems \ref{thm:SDR} and \ref{thm:mathieu}}

Let us start with some properties of maximally compact sequences.

\subsection{Properties of Maximally Compact Sequences}
\label{subsec:properties-MCS}
In the definitions of time and frequency spreads in Table \ref{tab:TF-disc-2} we considered complex sequences and their DTFTs. In the following, we establish two lemmas that make the search for maximally compact sequences easier. In the first lemma (Lemma \ref{lem:real}), we state that if someone gives us a complex sequence with a given time-frequency spread, we can always take the modulus and have a sequence with a smaller or equal time-frequency spread. This enables us to only consider {\em{real}} sequences as maximally compact sequences. In Lemma \ref{lem:shiftinv}, we state that shifts in time do not affect the time-frequency spread of sequences. This lemma allows us to assume---without loss of generality---that the sequences are centered around zero in time. 
\begin{Lemma}
  \label{lem:real}
  For any given sequence $x_k$,
  $$\eta_p(\abs{x}) \leq \eta_p(x)\,.$$
  
  \end{Lemma}
\begin{proof} 
 Let $x_k$ be a sequence with a given time spread $\Dn(x)$. Obviously,
  $$\Dn(x)\Eq \Dn(|x|).$$
Moreover, 
\begin{align*}
\Dwp(|x|)\Eq&\left|\sum_{k\in\Zspace}|x_k||x_{k+1}|\right|^{-2}-1\nonumber\\
\Leq&\left|\sum_{k\in\Zspace}x_k x_{k+1}^{\ast}\right|^{-2}-1\quad\Eq\Dwp(x).
\end{align*}
\end{proof}
Recall from Lemma \ref{lem:monotone} that for maximally compact sequences $\Dn$ changes monotonically with $\Dwp$. Thus, it is equivalent to fix either one of $\Dn$ or $\Dwp$, and optimize with respect to the other.

\begin{Lemma}
  \label{lem:shiftinv}
If $x$ is a maximally compact sequence, then $x_{k -\mu_n(x)}$ is also maximally compact. For non-integer $\mu_n$,  $x_{k-{\mu_n}}$ is a
shorthand for sinc resampling  on a grid shifted by $\mu_{n}$ in the time domain.

\end{Lemma}
\begin{proof} As $\mu_n$ is not necessarily an integer, we can use the Parseval's equality to show that the time center and spread is not affected by arbitrary shifts. These easy computations are left to the interested reader (They can be found in \cite{parhizkar2013}). Thus, if $x$ is a maximally compact sequence, then $x_{k -\mu_n(x)}$ is also maximally compact (note that time shift does not change the frequency characteristics of the sequence).

\end{proof}
\begin{Remark} Using Lemmas \ref{lem:real} and \ref{lem:shiftinv}, we only consider real sequences $x$, with $\mu_{n}(x)=0$ and $\|x\|_2=1$. Later in Theorem \ref{thm:mathieu} we show that maximally compact sequences are Mathieu's cosine functions, which have strictly positive inverse Fourier transforms. This enables us to state that maximally compact sequences are real sequences up to a shift, scale or modulation. 
\end{Remark}
\subsection{Proof of Theorem \ref{thm:SDR}}
\label{sec:proof-SDR}
By using Lemmas \ref{lem:real} and \ref{lem:shiftinv}, we can write  problem \eqref{eq:opt1} as
\begin{equation}\label{eq:opt}
\begin{aligned}
\Delta_{n,\text{opt}}^2 \ = \ &\underset{x_{n}}{\text{minimize}}
& & \sum_{k\in\Z} k^2 x_k^2 \\
& \text{subject to}
& &  \sum_{k\in\Z} x_k x_{k+1} = \frac{1}{\sqrt{1+\sigma^{2}}},\\
 & & &  \sum_{k\in\Z} x_k^{2} = 1.
\end{aligned}
\end{equation}
We can rewrite \eqref{eq:opt} in matrix form as a quadratically constrained quadratic program (QCQP)~\cite[p. 152]{boyd2004}:
\begin{equation}
\label{eq:matForm}
\begin{aligned}
& \underset{\x}{\text{minimize}}
& & \x^T\A\x \\
& \text{subject to}
& & \x^T\B\x  = \alpha\,,\\
&&& \x^T\x = 1\,,
\end{aligned}
\quad\quad\quad\quad \equiv \quad\quad\quad\quad 
\begin{aligned}
& \underset{\x}{\text{minimize}}
& & \text{tr}(\A\x\x^T) \\
& \text{subject to}
& & \text{tr}(\B\x\x^T)  = \alpha\\
&&& \text{tr}(\x\x^T) = 1\,.
\end{aligned}
\end{equation}
where $\A$ and $\B$ are defined in \eqref{eq:AandB} and $\alpha = 1/\sqrt{1+\sigma^2}$. Replacing $\x\x^T$ by $\X$, we can write equivalently
\begin{equation}
\begin{aligned}
& \underset{\X}{\text{minimize}}
& & \text{tr}(\A\X) \\
& \text{subject to}
& & \text{tr}(\B\X)  = \alpha\\
&&& \text{tr}(\X) = 1\\
&&& \X \succeq 0, \; \text{rank}(\X) = 1\,.
\end{aligned}
\label{eq:SDP3}
\end{equation}
We further relax the above formulation to reach the semi-definite program
\begin{equation}
\begin{aligned}
& \underset{\X}{\text{minimize}}
& & \text{tr}(\A\X) \\
& \text{subject to}
& & \text{tr}(\B\X)  = \alpha\\
&&& \text{tr}(\X) = 1,\ \ \X \succeq 0\,.
\end{aligned}
\label{eq:SDR2}
\end{equation}
In Lemma \ref{lem:tight} we show that the semi-definite relaxation is tight. This finishes the proof for Theorem \ref{thm:SDR}.
\qed
\begin{Lemma}
\label{lem:tight}
The semi-definite relaxation (SDR) from \eqref{eq:SDP3} to \eqref{eq:SDR2} is tight.
\end{Lemma}
\begin{proof}
Shapiro and then Barnivok and Pataki \cite{Shapiro1982,Barvinok1995,Pataki1998,Luo2010} show that if the SDP in \eqref{eq:SDP3} is feasible, then 
\begin{equation}
\text{rank}(\X^{\mathrm{opt}})\leq \lfloor(\sqrt{8m+1}-1)/2\rfloor\,,
\label{eq:rankSDR}
\end{equation}
where $m$ is the number of constraints of the SDP and
$\X^{\mathrm{opt}}$ is its optimal solution. For our semi-definite
program in \eqref{eq:SDR2}, $m = 2$. Thus, \eqref{eq:rankSDR} implies
that the solution has rank 1. Using this fact, one can see that the
semi-definite relaxation is in fact tight. Recall from Remark \ref{rem:feasible} that problem \eqref{eq:SDR} (and also \eqref{eq:SDP3}) is feasible. 
\end{proof}

\subsection{Proof of Theorem \ref{thm:mathieu}}
\label{sec:proof-mathieu}
We start by problem \eqref{eq:SDR} and its dual \eqref{eq:dual}.
\begin{Lemma}
\label{lem:dual}
For the primal problem \eqref{eq:SDR} and the dual \eqref{eq:dual}, strong duality holds.
\end{Lemma}
\begin{proof}
For a semi-definite program and its dual, if the primal is feasible and the dual is strictly feasible, then strong duality holds \cite{todd2001, Trnovska2005} . 

We saw in Remark \ref{rem:feasible} that the primal problem \eqref{eq:SDR}  is feasible. For the dual, one can use the Gershgorin's circle theorem and show that a sufficient condition for $\A - \lambda_1\,\B - \lambda_2\,\bm{I} \succ 0$ to hold is $\lambda_2 < -\lambda_1$ and $\lambda_1 >0$. Thus, the dual problem is strictly feasible.
\end{proof}
Thus, for finding the time-frequency spread of maximally compact sequences, solving the dual problem suffices. If a sequence is a solution to the dual SDP problem \eqref{eq:dual},
the dual constraint is active. Therefore, maximally compact sequences
lie on the boundary of the quadratic cone
$$ \A - \lambda_1\,\B - \lambda_2\,\bm{I} \succeq 0.$$
A maximally compact sequence $\x$ is thus solution of the eigenvalue problem
\begin{equation}
\label{eq:evboundary}
(\A - \lambda_1\,\B)\x \Eq \lambda_2\,\x,
\end{equation}
where $\lambda_{1}$ and $\lambda_{2}$ are the dual
variables of the SDP problem. $\lambda_2$ is also the minimal eigenvalue of $\A -
\lambda_1\,\B$ with $\x$ the associated eigenvector (this can be also seen by forcing the derivative of the Lagrangian in \eqref{eq:matForm} to zero).

This explicit link between the dual variables and the sequence,
yields a differential equation for which the DTFT spectrum of
maximally compact sequences is the solution. In the DTFT domain
\eqref{eq:evboundary} becomes (expanding the matrix multiplications)
\begin{align}
X^{\prime\prime}(e^{j\w})+\left(\lambda_2+\lambda_1\cos(\w)\right)X(e^{j\w})\Eq 0\:, \label{eq:mathieudiff}
\end{align}
which is Mathieu's differential equation \eqref{eq:mathieueq}. Taking into account the periodicity of \eqref{eq:mathieudiff}, it
appears that not all pairs of parameters $(a,q)$ will lead to a periodic
solution. Mathieu's functions can be restricted to be $2\pi$ periodic. The solutions of Mathieu's harmonic differential equation---equation
  \eqref{eq:mathieueq} with a $2\pi$-periodic solution $y$---are defined as
\begin{align}
\label{eq:mathieufun}
\textrm{Mathieu's harmonic Cosine (even, periodic) }\quad& \ce_m(q\:;\:\w)
\Eq \ce(a_m(q),q\:;\:\w)\ ,\ m\in\mathbb N.\\
\textrm{Mathieu's harmonic Sine (odd, periodic) }\quad&
\se_m(q\:;\:\w) \Eq \ce(b_m(q),q\:;\:\w)\ ,\ m\in\mathbb N^{+}.
\end{align}

It is immediately visible that the spectrum of maximally compact sequences
may only have  the form
\begin{equation}
\label{eq:mathieugen}
X(e^{j\w}) \Eq \begin{cases} \gamma_{0}\,\ce_{m}(-2\lambda_1\:;\:\w/2)+
\gamma_1\,\se_{m}(-2\lambda_1\:;\:\w/2)&\text{ for } m\in \mathbb N^{+},\\
\gamma_{0}\,\ce_{m}(-2\lambda_1\:;\:\w/2)&\text{ for } m=0\:,
\end{cases}
\end{equation}
for any constants $\gamma_0$ and $\gamma_1$ such that
$\|X(e^{j\w})\|=2\pi$. More specifically, for any $\lambda_{1}\geq 0$, the dual SDP problem can
be posed and any solution would have the form \eqref{eq:mathieugen}.

Charasteristic numbers of Mathieu's equation are ordered
\cite[p. 113]{bateman}, such that for $\lambda_1 >0$,
$$\begin{array}{rl} a_{0}(-2\lambda_1)< a_{1}(-2\lambda_1)<b_{1}(-2\lambda_1) <
b_{2}(-2\lambda_1)<a_{2}(-2\lambda_1)< \cdots\,.\end{array}$$

By (\ref{eq:mathieueq}) and with the substitution $\w\rightarrow \w/2$
one obtains $a_{m}(-2\lambda_1)=4\lambda_{2}$. Because $\lambda_{2}$ is the
minimal eigenvalue, we conclude that $m=0$.

Note that this result validates the one in \cite{SongGoh2002} which
stated that asymptotically Mathieu's functions minimize the
time-frequency product. With a diferent approach, we can establish
that only Mathieu's harmonic cosine of order 0 minimizes this product
for any given frequency-spread.

\section{Proof of Theorem \ref{thm:lower-bound}}
\label{sec:proof_low_up}
Let us start by proving the lower bound. 
\subsection{Lower Bound \eqref{eq:lower-bound}}
\label{sec:proof-lower-bound}
In order to prove the lower bound \eqref{eq:lower-bound} we first provide the following two lemmas.

\begin{Lemma}
\label{lem:ind}
Consider the sequence $a_k$ as follows with $\theta, \nu >0$
\begin{equation}
\label{eq:indfor}
a_{k+1} = (k+1)^2 + \theta - \frac{\nu}{a_k}\,.
\end{equation}
The sequence $a_k$ is positive for $k \geq k_0$ as soon as $a_{k_0+1} \geq a_{k_0} > 0$.
\end{Lemma}
\begin{proof}
Note that
\begin{equation}
\begin{aligned}
a_{k+2} - a_{k+1} &= 3+ 2k + \nu \frac{a_{k+1} - a_k}{a_{k+1}a_k}\,.
\end{aligned}
\end{equation}
Thus, by induction, the sequence is positive as soon as $a_{k_0+1} \geq a_{k_0} > 0$ for some $k_0$.
\end{proof}

\begin{Lemma}
If 
\begin{equation}
\lambda_2 < 1-\sqrt{1+\lambda_1^2}\,,
\label{eq:posCond}
\end{equation}
then $\P = \A - \lambda_1\,\B - \lambda_2\,\bm{I} \succeq 0$. 
\end{Lemma}
\begin{proof}
Note that if $\lambda_1 =0$, then the matrix $\P$ is positive-definite with the given condition. In the proof we will thus assume that $\lambda_1 \neq 0$. 

Consider the following tri-diagonal matrices $\P_1$ and $\P_2$:
\begin{equation}
\P_1 = \begin{bmatrix}[1.5]
1-\lambda_2	&	-\lambda_1/2	&				&	\bm{0}	\\
-\lambda_1/2	&	4 - \lambda_2	&	-\lambda_1/2	&			\\
			&	-\lambda_1/2	&	9 - \lambda_2	&	\ddots 	\\
\bm{0} 		&				&	\ddots		& \ddots
\end{bmatrix}, \, \quad \P_2 = \begin{bmatrix}[1.5]
0-\lambda_2	&	-\lambda_1	&				&	 \bm{0} \\
-\lambda_1/2	&	1-\lambda_2	&	-\lambda_1/2	&		\\
			&	-\lambda_1/2	&	4 - \lambda_2	&	\ddots\\
\bm{0}		&				&	\ddots		&	\ddots
\end{bmatrix}\,.
\end{equation}
Call $\mathcal{I}$ the set of eigenvalues of $\P$, $\mathcal{I}_1$ the
set of eigenvalues of $\P_1$ and $\mathcal{I}_2$  the set of
eigenvalues of $\P_2$. It is trivial to see that $\mathcal{I} = \mathcal{I}_1 \cup \mathcal{I}_2$.

We show that if condition \eqref{eq:posCond} is satisfied, then both $\P_1$ and $\P_2$ have positive eigenvalues.

\begin{enumerate}
\item 
$\P_1$: Sylvester's criterion states that a symmetric matrix is positive definite if and only if its principal minors are all positive\footnote{For an infinite matrix that can be regarded as the matrix of a bounded operator in $\ell^2$, one can show that Sylvester's criterion still holds, i.e., the quadratic form is non-negative definite if and only if all the finite principal minors of the matrix are positive. Note that we consider $\norm{x}_2 = 1$, thus $x \in \ell^2$.}. As $\P_1$ is symmetric, we can use Sylvester's criterion on it. We use Gaussian elimination on the matrix $\P_1$ to compute its principal minors
\begin{equation}
{\P}_1^{\text{U}} = \begin{bmatrix}[1.5]
1-\lambda_2	&	-\lambda_1/2	&				&				&	\bm{0}	\\
0			&		s_2		&	-\lambda_1/2	&				&		\\
			&		0		&	s_3			&	-\lambda_1/2	&		\\
\bm{0}			&				&	0			&	\ddots		&	\ddots 	
\end{bmatrix}\,,
\end{equation}
where
\begin{equation*}
s_{k+1} = (k+1)^2 - \lambda_2 - \frac{\lambda_1^2}{4s_k}\,, \quad k \geq 1\,.
\end{equation*}
This satisfies the induction formula in \eqref{eq:indfor}. Thus according to Lemma \ref{lem:ind}, the principal minors of ${\P}_1^{\text{U}}$ (and so $\P_1$) are positive as soon as $s_2 \geq s_1 > 0$. This is equivalent to $1-\lambda_2 > 0$ and $4- \lambda_2 - \frac{\lambda_1^2}{4(1-\lambda_2)} \geq 1-\lambda_2$. i.e., $\lambda_1^2 \leq 12(1-\lambda_2)$, which is a weaker condition than \eqref{eq:posCond}.
Therefore, $\P_1$ is also positive semi-definite.
$$\:$$
\item $\P_2$: We can decompose $\P_2$ as
\begin{equation}
\P_2 = \begin{bmatrix}
2			&			&				&	& \bm{0}\\
			&	1		&				&	& 			\\
			&			&	1			&	& 			\\
\bm{0}	&			&	 			&	&  \ddots\\
\end{bmatrix} \times \begin{bmatrix}[1.5]
0-\lambda_2/2	&	-\lambda_1/2	&				&	 \bm{0} \\
-\lambda_1/2	&	1-\lambda_2	&	-\lambda_1/2	&		\\
			&	-\lambda_1/2	&	4 - \lambda_2	&	\ddots\\
\bm{0}		&				&	\ddots		&	\ddots	
\end{bmatrix}\ = \ \D \times \P_2^\text{(s)} \ .
\end{equation}
Note that both $\D$ and $\P_2^\text{(s)}$ are symmetric. Also observe that the eigenvalues of $\P_2$ and $\P_2^{\text{(s)}}$ are equal. 
Thus, it suffices to consider the eigenvalues of
$\P_2^\text{(s)}$; If $\P_2^\text{(s)}$ is positive-definite then all the eigenvalues of
$\P_2$ are positive. 
%

Again using Gaussian elimination on $\P_2^\text{(s)}$ results in
\begin{equation}
\P_2^{\text{(s),U}} = \begin{bmatrix}[1.5]
-\lambda_2/2	&	-\lambda_1/2	&		0		&				&	\bm{0}	\\
0			&	s_1			&	-\lambda_1/2	&				&		\\
			&	0			&	s_2			&	-\lambda_1/2	&		\\
\bm{0}		&				&	0			&	\ddots		&	\ddots 	
\end{bmatrix}\,,
\end{equation}
where $s_1 = 1 - \lambda_2 +\lambda_1^2/2\lambda_2$ and $s_k$ has the following form $s_{k+1} = (k+1)^2 - \lambda_2 - \frac{\lambda_1^2}{4s_k}\,, \quad k \geq 1$. This again satisfies the induction formula \eqref{eq:indfor}. Therefore, it suffices to show that $\lambda_2 < 0$ and $s_2 \geq s_1 > 0$. 
We show that under condition  \eqref{eq:posCond}, this is true. In order to satisfy $s_1 > 0$, we need  $1- \lambda_2 + \lambda_1^2/2\lambda_2 >0$, which is equivalent to $\lambda_2 < \frac{1}{2}(1-\sqrt{1+2\lambda_2^2})$. This is a weaker condition than \eqref{eq:posCond}. 

Further, in order to satisfy $s_2 > s_1$, we need
$$s_2 - s_1 = 3 + \frac{\lambda_1^2}{4s_1 \lambda_2^2}(\lambda_2^2 - 2\lambda_2-\lambda_1^2) \geq 0\,.$$ 
Thus, it is enough to have $\lambda_2^2 - 2\lambda_2-\lambda_1^2 \geq 0$. That is $\lambda_2 \leq 1-\sqrt{1+\lambda_1^2}$, which is the bound provided in \eqref{eq:posCond}. Putting these together, we can conclude that under condition \eqref{eq:posCond}, the matrix $\P$ is positive-definite. 
\end{enumerate}
\end{proof}
Note that condition \eqref{eq:posCond} gives a sufficient (but not
necessary) condition on the feasible set of the dual problem
\eqref{eq:dual}. Thus, it provides a lower bound for the maximum value of the dual. 
Consider the restricted dual problem 
\begin{equation}
\begin{aligned}
& \underset{\lambda_1, \lambda_2}{\text{maximize}}
& & \alpha\,\lambda_1 + \lambda_2 \\
& \text{subject to}
& & \lambda_2 < 1-\sqrt{1+\lambda_1^2}
\end{aligned}
\label{eq:dualres}
\end{equation}

The solution to this problem is simply $1-\sqrt{1-\alpha^2}$. If we rewrite $\alpha$ in terms of $\sigma^2$, we finally have $$\Delta_{n,\text{opt}}^2 \geq 1- \sqrt{\frac{\sigma^2}{1+\sigma^2}}\,.$$
This concludes the proof for the lower bound.

\subsection{Upper Bound \eqref{eq:upper-bound}}
\label{sec:proof-upper-bound}

It is easy to see that for small values of $\sigma^2$ (equivalently $\alpha$ closer to $1$), the maximum of the dual problem is achieved for large values of $\lambda_1$ (remember that $\lambda_1$ needs to be positive). We saw in \ref{sec:proof-mathieu} that for maximally compact sequences (i.e.~sequences that result in the maximum of the dual problem) we have $\lambda_2 = 1/4\ a_0(2\lambda_1)$ (note that $a_0(-q) = a_0(q)$). McLachlan in \cite{McLachlan64} shows that for large enough values of $q$, we have
\begin{equation}
\begin{aligned}
a_0(q) & = -2q + 2 q^{\frac{1}{2}} - \frac{1}{4} - \frac{1}{32} q^{\frac{-1}{2}}  - \frac{48}{2^7}q^{-1} - \frac{848}{2^{17}} q^{\frac{-3}{2}} - \frac{4,752}{2^{20}} q^{-2}- \frac{126,752}{2^{20}}q^{\frac{-5}{2}} - \cdots\\
& = -2q + 2 q^{\frac{1}{2}} - \frac{1}{4} - \frac{1}{32} q^{\frac{-1}{2}} + O(q^{-1})\,.
\end{aligned}
\end{equation}
Thus, we have for large $q$,
\begin{equation}
\label{eq:a0-upper}
a_0(q) \leq -2q + 2 q^{\frac{1}{2}} - \frac{1}{4}\,.
\end{equation}

\begin{figure}[tb]
\centering
\includegraphics[width=0.75\linewidth]{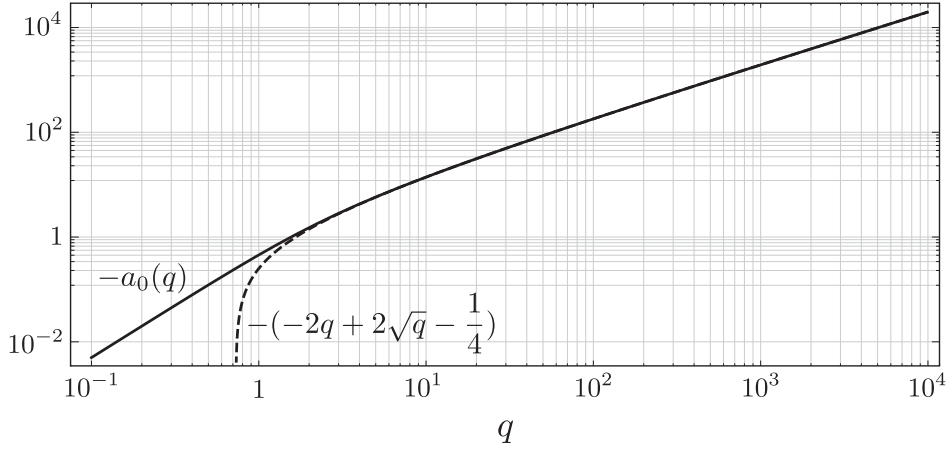}
\caption{The characteristic number $a_0(q)$ of Mathieu's function together with its upper bound. In order to plot on a log scale, the negations are plotted.}
\label{fig:charFig}
\end{figure}

The values for $a_0(q)$ and its upper bound are shown in Figure \ref{fig:charFig}. In order to be able to plot on a log scale, their negative counterparts are plotted. 
 
After replacing $q$ by $2\lambda_1$ and $a_0(2\lambda)$ by $4\lambda_2$ in \eqref{eq:a0-upper}, we have
\begin{equation}
\lambda_2  = \frac{1}{4} a_0(2\lambda_1) \leq -\lambda_1 + \frac{1}{\sqrt{2}}\sqrt{\lambda_1}-\frac{1}{16}\,.
\end{equation}
Because this set contains the original feasible set of the dual problem, it will give an upper bound on the optimal value of the dual. In other words, $\Dn \leq \sigma_r^2$, where 
\begin{equation}
\begin{aligned}
\sigma_r^2 \Eq & \underset{\lambda_1, \lambda_2}{\text{maximize}}
& & \alpha\,\lambda_1 + \lambda_2 \\
& \text{subject to}
& & \lambda_2 \leq -\lambda_1 + \frac{1}{\sqrt{2}}\sqrt{\lambda_1}-\frac{1}{16}\,.
\end{aligned}
\label{eq:dualrelaxed}
\end{equation}

It is easy to see that the maximum of \eqref{eq:dualrelaxed} is achieved for $\lambda_1 = \frac{1}{8(1-\alpha)^2}$. Replacing $\lambda_1$ in \eqref{eq:dualrelaxed} and using the fact that $\Delta_{n, \text{opt}}^2 \leq \sigma_r^2$, leads to 
\begin{equation}
\Delta_{n, \text{opt}}^2 \leq \frac{1}{8}\left(\frac{\sqrt{1+\sigma^2}}{\sqrt{1+\sigma^2}-1}-\frac{1}{2}\right)\,.
\end{equation}

\section{Simulation Results}
\label{sec:simulation-results}

In order to show the behaviour of the results obtained in Theorems \ref{thm:SDR} and \ref{thm:lower-bound}, we ran some simulations. For this, we assumed that the designed filter is finite length with 201 taps in the time domain. The length is long enough not to pose restrictions on the solution for the considered frequency spreads. For smaller frequency spreads, we can increase the length of the sequence. For different values of $\Dwp \Eq \sigma^2$, we solved the semi-definite program \eqref{eq:SDR} using the \textsf{cvx} toolbox in MATLAB. 

The resulting values of $\Dn$ were then multiplied with the
corresponding $\Dwp$ to produce the time-frequency spread of maximally
compact sequences. The time-frequency spread of maximally compact
sequences versus their frequency spread is shown with the solid curve
in Figure \ref{fig:bounds}. This means---numerically---that any
time-frequency spread under this curve is not achievable. The dotted
line in this figure shows the classic Heisenberg uncertainty
bound. Comparing the two curves shows the gap between the classic
Heisenberg principle and what is achievable in practice. The dashed
lines represent analytical lower and upper bounds for the
time-frequency spread of maximally compact sequences (found in Theorem \ref{thm:lower-bound}).

\begin{figure}[tb]
\centering
\scalebox{.9}{
    \footnotesize
    \psfrag{i}[Bl][Bl]{Infeasible TF region (analytically proven)}
    \psfrag{I}[Bl][Bl]{Infeasible TF region (numerically proven)}
    \psfrag{s}[Br][Br]{Infimum TF spread by tight SDR}
    \psfrag{l}[Br][Br]{Analytic lower bound}
    \psfrag{u}[Bl][Bl]{Analytic upper bound}
    \psfrag{h}[Br][Br]{Heisenberg lower bound}
    \psfrag{f}[Bc][Bc]{\small{$\Dwp$}}
    \psfrag{t}[Bc][Bc]{\small{$\etap\:=\:\Dwp\times\Dn$}}
\psfragfig[width=1\textwidth]{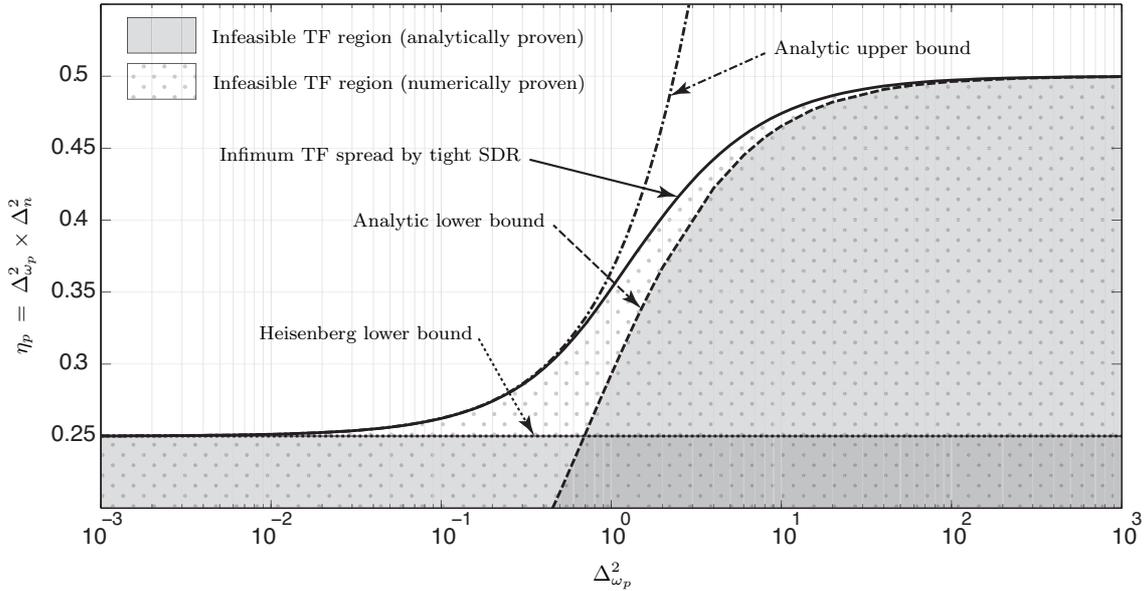}}
\caption{\textbf{New uncertainty bounds}. The solid line shows the results of solving the SDP in \eqref{eq:SDR}. The dotted line shows the classic Heisenberg uncertainty principle. The dashed lines show the analytic lower and upper bounds found in Theorem \ref{thm:lower-bound}.}
\label{fig:bounds}
\end{figure}
Further, to give an insight on how the time-frequency spread of some
common filters compare to that of maximally compact sequences, we plot
their time-frequency spread together with the new uncertainty bound in
Figure \ref{fig:windows}. By changing the length of each filter in
time, we can find its time and frequency spreads which results in a
point on the figure. We observe that as shown by Prestin et al.~in
\cite{Prestin2003}, asymptotically when the frequency spread of
sequences are very small, sampled Gaussians converge to the lower bound for maximally compact sequences. 

\begin{figure}[tb]
\centering
\includegraphics[width = 0.9\linewidth]{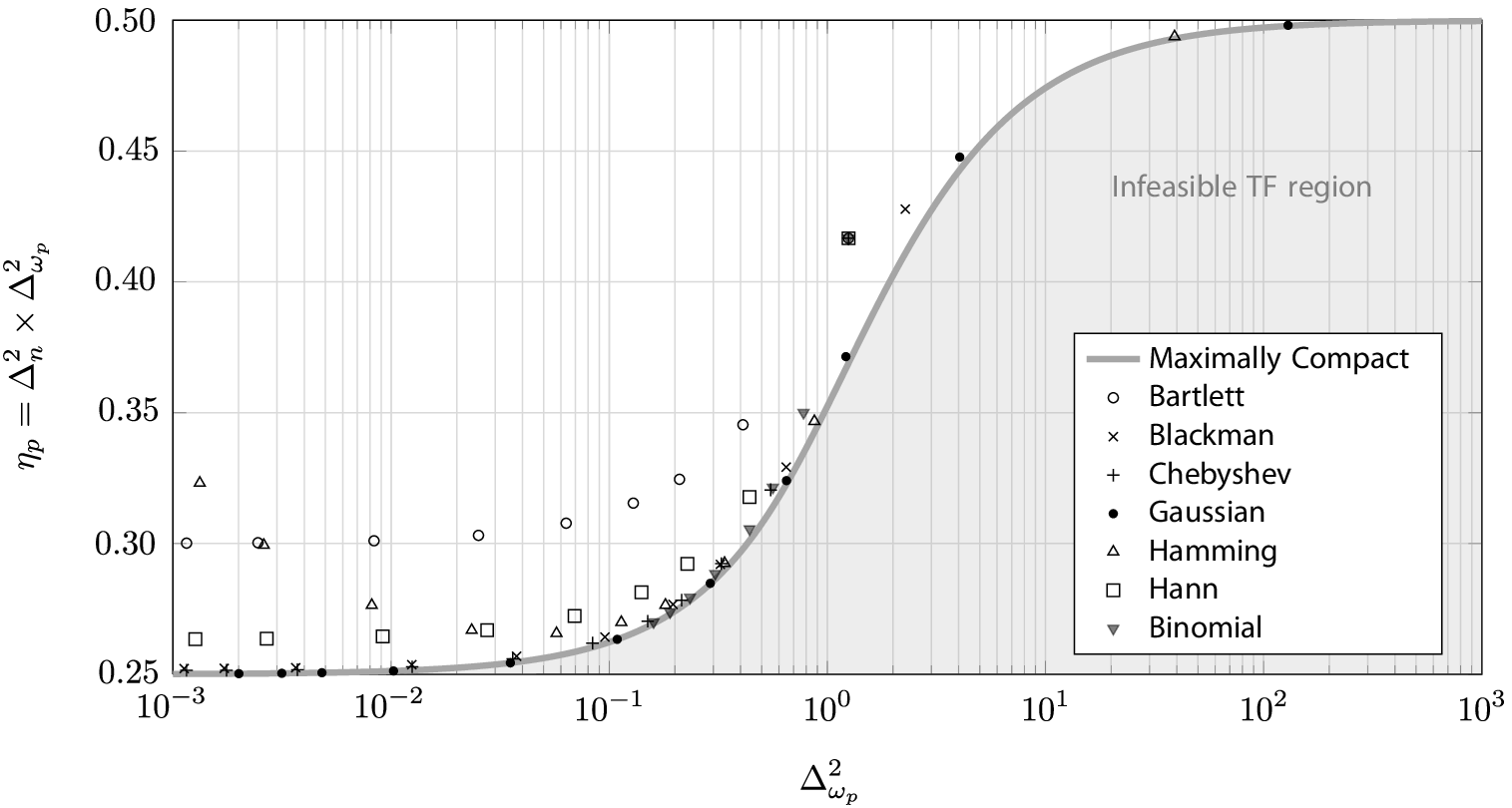}
\caption{\textbf{Time-frequency spread of common FIR filters.} By changing the length of the filters in time, we compute the time and frequency spreads for each type of filters. For small values of the frequency spread, Gaussian filters are good approximations of Mathieu's functions (as shown also in \cite{Prestin2003}).}
\label{fig:windows}
  \end{figure}

\section{Conclusion}
\label{sec:conclusion}

We showed that for discrete-time sequences, contrary to continuous-time signals, the classic Heisenberg uncertainty bound is not always achievable and the uncertainty minimizers have a large gap from this bound. We constructed an optimization framework for finding the uncertainty minimizers which we refer to as the ``maximally compact sequences''. This framework allows one to find---numerically and efficiently---the most compact sequence in time with a desired frequency spread. We further showed that the discrete-time Fourier transform of these sequences is a very special class of Mathieu's functions. We also proved analytic bounds on the time-frequency spread of discrete sequences. Maximally compact sequences can serve as optimal probing windows (similarly to the ones shown in Figure \ref{fig:windows}) in several signal processing applications. Furthermore, the connection provided between the solutions of the semi-definite program and Mathieu's functions, also enables the approximation of Mathieu's functions with arbitrary accuracy. 

\section*{Ackowledgements}
 The authors would like to thank Dr.~Arash Amini and Dr.~Seyed Hamed Hassani for their help in
 proving Theorem \ref{thm:lower-bound}, and Ivan Dokmanic for his
 insight on Mathieu's functions. The authors would also like to thank the reviewers for their valuable contributions in the improvement of the manuscript, and pointing out a related recent work \cite{Nam2013}, which was published on arxiv after the submission of this manuscript and during the review process.

\bibliographystyle{alpha}
\bibliography{references_ACHA}
\end{document}